\documentclass[11pt,a4paper,twoside]{article}
\usepackage[utf8]{inputenc}
\usepackage{amsmath,amssymb,amsthm}
\usepackage{multirow,array,algorithm,algorithmic}
\newtheorem{theorem}{Theorem}
\usepackage[authoryear]{natbib}
\usepackage[left=2cm,right=2cm,top=2.5cm,bottom=2.5cm]{geometry}
\usepackage{xcolor,graphicx,setspace}
\usepackage{authblk}

\newcommand{\se}[1]{$\scriptstyle(#1)$}
\newcommand{\best}[1]{\underline{#1}}
\newcolumntype{H}{>{\setbox0=\hbox\bgroup}c<{\egroup}@{}}

\usepackage{url}
\usepackage[colorlinks,linkcolor=blue,citecolor=blue,urlcolor=blue]{hyperref}

\begin{document}
\title{\Large\bf Nearest Neighbor Classification Based on Imbalanced Data:\\ A Statistical Approach}

\author[1]{Anvit Garg}
\author[2]{Anil K. Ghosh}
\author[3]{Soham Sarkar}

\affil[1]{{Department of Statistics, Harvard University, \protect\\ One Oxford Street, Cambridge - 02138, MA, USA. Email: \href{mailto:anvitgarg@fas.harvard.edu}{anvitgarg@fas.harvard.edu}}\vspace{0.2in}} 
\affil[2]{{Theoretical Statistics and Mathematics Unit, Indian Statistical Institute, \protect\\ 203, B.\ T.\ Road, Kolkata 700108, India. Email: \href{mailto:akghosh@isical.ac.in}{akghosh@isical.ac.in}}\vspace{0.2in}}
\affil[3]{{Theoretical Statistics and Mathematics Unit, Indian Statistical Institute, Delhi Centre, \protect\\ 7, S.\ J.\ S.\ Sansanwal Marg, New Delhi 110016, India. Email: \href{mailto:sohamsarkar@isid.ac.in}{sohamsarkar@isid.ac.in}}}

\date{}

\maketitle

\markboth{Garg \MakeLowercase{\textit{et al.}}:}{Nearest Neighbor Classification Based on Imbalanced Data}

\begin{abstract}
When the competing classes in a classification problem are not of comparable size, many popular classifiers exhibit a bias towards larger classes, and the nearest neighbor classifier is no exception. To take care of this problem, we develop a statistical method for nearest neighbor classification based on such imbalanced data sets. First, we construct a classifier for the binary classification problem and then extend it for classification problems involving more than two classes. Unlike the existing oversampling or undersampling methods, our proposed classifiers do not need to generate any pseudo observations or remove any existing observations, hence the results are exactly reproducible. We establish the Bayes risk consistency of these classifiers under appropriate regularity conditions. Their superior performance over the existing methods is amply demonstrated by analyzing several benchmark data sets. 

\end{abstract}

\section{Introduction}
\label{sec:intro}

Nearest neighbor classifier \citep[see][]{cover1967,friedman2001,duda2000} is arguably the simplest and most popular nonparametric classifier in statistics and machine learning literature. It classifies a test case ${\bf x}$ to the class which has the maximum representation of training data points in a neighborhood of ${\bf x}$. So, when the numbers of training sample observations from the competing classes are not comparable, it shows a tendency to classify more observations in favor of larger classes. Many other classifiers (e.g., classification tree, random forest, support vector machines, artificial neural networks) also have similar problems when the training sample is imbalanced.

One possible way of dealing with such  imbalanced datasets is to assign different weights or costs for observations belonging to different classes \citep[e.g.,][]{huang2005,dubey2013,zang2016}. But these choices of weights are somewhat adhoc, and the performance of the resulting classifier may depend heavily on the choice. Another popular approach in the machine learning literature is to balance the training sample \citep{domingos1999}, either by under-sampling from the majority class \citep[e.g.,][]{hart1968,tomek1976b,wilson1972,tomek1976a,hart1968,kubat1997,laurikkala2001,smith2014,mani2003,mishra2017,dittman2014} or by oversampling from the minority class \citep[e.g.,][]{chawla2002,han2005,he2008,nguyen2011,menardi2014,last2017,das2020}. The under-sampling approach removes some data points from the majority class, but this can also lead to the removal of important observations and negatively affect the fit of the majority class. The oversampling techniques, on the other hand, synthetically generate observations from the minority class.  Some of these methods need to generate data every time we want to classify a new observation, and some of them need a higher memory space to store all over-sampled points. Clearly, these are not desirable for large datasets, especially when the proportion of observations in the minority class is small. Further, two different persons using the same oversampling method on the same data set may get widely different results due to the randomness involved in sample generation.

In this article, we take a statistical approach towards nearest neighbor classification based on imbalanced data. Our method does not need any adhoc weight or cost function for its implementation. Also unlike the data balancing methods, it does not need to delete some observations from the majority class or generate additional observations from the minority class.  This method is based on a probabilistic model, and the results are exactly reproducible. In the next section, we describe our method for binary classification. We establish its consistency under appropriate regularity conditions, and analyze some simulated and real datasets to compare its performance with the usual nearest neighbor classifier, a weighted nearest neighbor classifier, and various data-balancing algorithms. In Section~\ref{sec:multi-class}, we extend our method for classification problems involving more than two classes. Consistency of the resulting classifiers is derived and some real data sets are analyzed to investigate their empirical performance. Finally, Section~\ref{sec:concluding_remarks} provides a brief summary of the work and some concluding remarks. 

\section{The proposed binary classifier}
\label{sec:method}

One major problem with the nearest neighbor classification based on imbalanced data is the sparsity of the minority class observations. Sometimes, one may not get any training data point from the minority class in a neighborhood of a query point even if the true posterior probability of that class is high at that location. Assignments of different weights or costs to observations belonging to different classes are not of much help in such situations. To take care of this problem, unlike the usual nearest neighbor classifier, here we do not consider the same number of neighbors for all test cases. For a test point ${\bf x}$, we keep on considering its neighbors one by one according to their distances until we get a fixed number of neighbors ($k$, say) from the minority class. Let $p({\bf x})$ be the probability (which can be assumed to be constant over a small neighborhood of ${\bf x}$) that a randomly chosen neighbor of ${\bf x}$ comes from the minority class. If $N_k({\bf x})$ denotes the minimum number of neighbors needed to get $k$ neighbors from the minority class, one can see that $N_k({\bf x})$ follows a negative binomial distribution \citep[e.g.,][]{johnson2005} with the probability mass function (p.m.f.) given by
\[
P(N_k({\bf x})=n)=\binom{n-1}{k-1} (p({\bf x}))^k (1-p({\bf x}))^{n-k},
\]
for $n=k,k+1,\ldots$ Note that if the majority class and the minority class have identical probability distributions, and there are $n_1$ and $n_2$ observations from these two classes, respectively, we have $p({\bf x})=n_2/(n_1+n_2)$ ($=p_0$, say) for all ${\bf x}$. In such a case, the p.m.f.\ is given by
\[
f_k(n)=\binom{n-1}{k-1} p_0^k(1-p_0)^{n-k}.
\]
The random variable $N_k({\bf x})$ is supposed to take a small (respectively, large) value if $p({\bf x})$ is large (respectively, small). If $N_0$ is an independent random variable with p.m.f. $f_k$, the probability $P(N_0< N_k({\bf x}))=\sum_{n<N_k({\bf x})}f_k(n)$ can be viewed as the $p$-value for testing the null hypothesis $H_0:p({\bf x})=p_0$ against the alternative hypothesis $H_1:p({\bf x})>p_0$ \citep[see, e.g.,][]{casella2021}. Clearly, lower $p$-values indicate stronger evidence in favour of the minority class. Similarly, $P(N_0> N_k({\bf x}))=\sum_{n>N_k({\bf x})}f_k(n)$ can be viewed as the $p$-value for testing $H_0:p({\bf x})=p_0$ against  $H_1:p({\bf x})<p_0$, and lower values of it indicate stronger evidence in favor of the majority class. Because of the discrete nature of the random variable $N_k({\bf x})$, here we make a slight adjustment and use 
\[
e_k{(\bf x)}=\sum_{n<N_k({\bf x})}f_k(n) +\frac{1}{2}f_k(N_k({\bf x}))
\]
and $1-e_k({\bf x})=\sum_{n>N_k({\bf x})}f_k(n) +\frac{1}{2}f_k(N_k({\bf x}))$ as adjusted $p$-values for these two cases. Under $H_0:p({\bf x})=p_0$, $e_k({\bf x})$ follows a distribution symmetric about $1/2$ (i.e., $e_k({\bf x})$ and $1-e_k({\bf x})$ have the same distribution). On the other hand, $e_k({\bf x})>1/2$ (respectively, $e_k({\bf x})<1/2$) gives an evidence in favour of the majority (respectively, minority) class, and a larger value of $|e_k({\bf x})-1/2|$ indicates stronger evidence. We compute $e_k({\bf x})$ for $k=1,2,\ldots,k_{\max}$ ($k_{\max}$ is a user specified value), and consider ${\cal E}_1= \max_{k \le k_{\max}} e_k({\bf x})$ as the strongest evidence in favour of the majority class (Class-1, say). Similarly,  ${\cal E}_2= \max_{k \le k_{\max}} (1-e_k({\bf x}))=1-\min_{k \le k_{\max}} e_k({\bf x})$ can be viewed as the strongest evidence in favour of the minority class (Class-2, say). We classify the observation ${\bf x}$ to Class-1 (respectively, Class-2)  if ${\cal E}_1 >{\cal E}_2$ (respectively, ${\cal E}_1 <{\cal E}_2$). The procedure is summarized below in Algorithm~\ref{algo:nn_imbalanced_binary}. Note that unlike the data balancing algorithms, no randomness is involved in this proposed method, and the results are completely reproducible.

\begin{algorithm}[h]
\caption{Proposed method for binary classification\label{algo:nn_imbalanced_binary}}
	\begin{algorithmic}[1]
		\renewcommand{\algorithmicrequire}{\textbf{Input:}}
		\renewcommand{\algorithmicensure}{\textbf{Output:}}
		\REQUIRE ${\bf x}$ -- Observation to be classified. \\ $\{{\bf x}_{11},\ldots,{\bf x}_{1n_1}$\}: Training data from majority class.\\ 
		$\{{\bf x}_{21},\ldots,{\bf x}_{2n_2}$\}: Training data from minority class.
		\ENSURE $\delta({\bf x})$: Predicted class label of ${\bf x}$.  
		\\ 
		\STATE  Sort the $n_1+n_2$ training data points in increasing order of distances from ${\bf x}$.\\
		\textit{Initialization} : $e_{\min}=e_{\max}=0.5$\\
		\FOR {$k=1$ to $k_{max}$}
		\STATE Compute $n_0$, the minimum number of neighbors of ${\bf x}$ needed to get $k$ neighbors from the minority class.
		\STATE Calculate $e_k=\sum_{n<n_0}f_k(n) +\frac{1}{2}f_k(n_0)$.
		\IF {($e_k<e_{\min}$)}
		\STATE $e_{\min} = e_k$.
		\ENDIF
		\IF {($e_k>e_{\max}$)} 
		\STATE $e_{\max} = e_k$.
		\ENDIF
		\ENDFOR
		\STATE Compute ${\cal E}_1=e_{\max}$ and ${\cal E}_2=1-e_{\min}$.  
		\RETURN $\delta({\bf x})=1+{\mathbb I}_{\{{\cal E}_1<{\cal E}_2\}}$
	\end{algorithmic}
\end{algorithm}
	
\subsection{Consistency}
\label{sec:binary_consistency}

Now, we investigate the large sample behavior of our proposed classifier. For this investigation, we assume that the density functions of the competing classes are continuous, and sample proportions of the two classes are asymptotically non-negligible. Our method needs the value of $k_{\max}$ to be specified. We assume that $k_{\max}$ grows with $n$ in such a way that $k_{\max}/n \rightarrow 0$ as $n \to \infty$. This ensures that the neighborhood around any query point ${\bf x}$ shrinks as $n$ increases so that we can capture the local behavior of the underlying densities. These assumptions are quite common in the literature of nearest neighbor methods \citep[see, e.g.,][]{loftsgaarden1965,cover1967}. 
Under these assumptions, we have the following result.

\begin{theorem}\label{thm:binary_consistency}
Suppose that the two competing classes have continuous density functions $f_1$ and $f_2$. Also assume that as $n \rightarrow \infty$, (i) $n_i/n$ converges to some $\pi_i \in (0,1)$ for $i=1,2$ ($\pi_1\ge \pi_2,~\pi_1+\pi_2=1$)  and (ii) $k_{\max} \rightarrow \infty$ in such a way that $k_{\max}/n \rightarrow 0$. If $f_i({\bf x})>f_j({\bf x})$ ($i,j=1,2,~i\neq j$), then the proposed method classifies the observation ${\bf x}$ to the $i$-th class with probability tending to $1$ as $n \to \infty$.
\end{theorem}
\begin{proof} Let ${\bf x}$ be the observation to be classified. First consider a sequence $\{k_n:n\ge 1\}$, where $k_n \rightarrow \infty$ and $k_n/n \rightarrow 0$ as $n \rightarrow \infty$. Define $r_{k_n}$ as the distance between ${\bf x}$ and its $k_n$-{th} nearest neighbor from Class-2. Let $B({\bf x},r_{k_n})$ be the closed ball (neighborhood) of radius $r_{k_n}$ around ${\bf x}$ and $N_{k_n}$ be the total number of observations in $B({\bf x}, r_{k_n})$. Also, define $\pi_1^{(n)}$ (respectively, $\pi_2^{(n)}$) as the probability that a random observation in $B({\bf x},r_{k_n})$ is from the first (respectively, second) class. One can see that $N_{k_n}$ follows the negative binomial distribution \citep[see, e.g.,][]{feller2008} with parameters $k_n$ and $\pi_2^{(n)}$. The evidence in favor of the first class is given by 
\[
e_{k_n}=\sum_{n<N_{k_n}} f_{k_n}(n) +\frac{1}{2}f_{k_n}(N_{k_n}).
\]
Clearly, $e_{k_n}$ is a function of $N_{k_n}$, and it can be viewed as the conditional probability $$\psi(N_{k_n})=P(N_0 < N_{k_n} \mid N_{k_n}) + \frac{1}{2}P(N_0 = N_{k_n}\mid{N_{k_n}}),$$ where $N_0$ is an independent negative binomial random variable with parameters $k_n$ and $p_0={n_2}/{n}$. 

Now, define $T_n=(N_0-N_{k_n})/k_n$. The mean and the variance of $T_n$ are 
$\mu_n = \frac{n_1}{n_2} - \frac{\pi_1^{(n)}}{\pi_2^{(n)}}$ and $\sigma^2_n = \frac{1}{k_n}\left[\frac{n n_1}{n_2^2} - \frac{\pi_1^{(n)}}{\left(\pi_2^{(n)}\right)^2}\right]$, respectively. Under the condition $k_n/n \rightarrow 0$, since $r_{k_n} \stackrel{P}{\rightarrow} 0$  \citep[see][]{loftsgaarden1965}, using the continuity of $f_1$ and $f_2$, 
for $j=1,2$, we have $\left|\pi_j^{(n)} - \frac{\pi_j f_j({\bf x})}{\pi_1 f_1({\bf x}) + \pi_2 f_2({\bf x})}\right| \stackrel{P}{\rightarrow} 0$ as $n\rightarrow \infty$. Again, $n_j/n \rightarrow \pi_j$ for $j=1,2$. So, $\mu_n$ converges to  $\frac{\pi_1}{\pi_2}-\frac{\pi_1 f_1({\bf x})}{\pi_2f_2({\bf x})}$ and $\sigma^2_n$ converges to $0$. This implies  $T_n \xrightarrow{P} \frac{\pi_1}{\pi_2}-\frac{\pi_1 f_1({\bf x})}{\pi_2f_2({\bf x})}$.  Note that this limiting value is positive (respectively, negative) if $f_1({\bf x}) < f_2({\bf x})$ (respectively, $f_1({\bf x}) > f_2({\bf x})$). Therefore,
\[
P(N_0 < N_{k_n})=P(T_n<0) \to 
\begin{cases}
	1 & \text{ if } f_1({\bf x}) > f_2({\bf x})\\
	0 & \text{ if } f_1({\bf x}) < f_2({\bf x}).
\end{cases}
\]
In both of these cases, $P(N_{k_n}=N_0) \rightarrow 0$ as $n\rightarrow \infty$. Note that $P(N_0<N_{k_n})+\frac{1}{2}P(N_0=N_{k_n})=E(\psi(N_{k_n}))$ So, $E(\psi(N_{k_n}))$ converges to $1$ (respectively, $0$) if $f_1({\bf x}) > f_2({\bf x})$ (respectively, $f_1({\bf x}) < f_2({\bf x})$). But $\psi(N_{k_n})$ is bounded between 0 and 1. So, $e_{k_n}=\psi(N_{k_n})$ also converges in probability to $1$ and $0$ in the respective cases. 

From the above discussion, it is clear that if $f_1({\bf x})>f_2({\bf x})$, the strongest evidence in favor of Class-1, ${\cal E}_1 = \max_{k\le k_{max}} e_k\ge e_{k_n}=\psi(N_{k_n})\rightarrow 1$ as $n \rightarrow \infty$. Now, we need to show that ${\cal E}_2=1-\min_{k\le k_{max}} e(k)$, the strongest evidence in favor of Class-2, remains bounded away from $1$ as $n$ diverges to infinity. If possible, assume this is not true. Then, there exists a sequence $\{k^*_n ~: n \ge 1\}$ such that $e_{k_n^*} \rightarrow 0$ as $n \rightarrow \infty$. Clearly, this is not possible if $k_n^{*}$ remains bounded (in that case, $f_{k^*_n}(t)$ remains bounded away from $0$ for all $t$). On the other hand, we have proved that if $k_n^{*}\rightarrow \infty$, $e_{k_n^*} \rightarrow 1$ (note that $k_n^* \le k_{\max}$ and hence $k_n^*/n \rightarrow 0$ as $n \rightarrow \infty$). So, ${\cal E}_2$ remains bounded away from $1$. As a result, we classify ${\bf x}$ to Class-1 with probability tending to one. Similarly, for $f_1({\bf x})<f_2({\bf x})$, we classify ${\bf x}$ to Class-2 with probability converging to $1$ as $n$ increases.
\end{proof}



\subsection{Illustrative examples}
\label{sec:binary_simulation}

We consider some simulated examples to demonstrate the utility of the proposed method. For each example, we consider training and test samples of size 1000. While the training samples have $\alpha$ ($0<\alpha<1/2$) proportion of observations from the minority class, the test samples are considered to have equal number of observations from the two classes. To evaluate the performance of our method, we compute the precision and the recall of the proposed classifier on the test set. Suppose that a classifier induces an allocation matrix as shown in Table~\ref{tab:allocation_matrix} on the test data. Then, precision and recall of the classifier for Class-$i$ ($i=1,2$) are defined as ${\cal P}(i)=n_{ii}/{n_{0i}}$ and ${\cal R}(i)=n_{ii}/{n_{i0}}$, respectively. Often the harmonic mean of precision and recall, called the $F_1$ score, is used as another measure of performance. The $F_1$ score for the Class-$i$ ($i=1,2$) is given by ${\cal F}_1(i)=2n_{ii}/(n_{i0}+n_{0i})$. We compute $\mathcal P(i), \mathcal R(i)$ and $\mathcal F_1(i)$ for $i=1,2$, and define the overall precision, recall and $F_1$ score of the classifier as ${\cal P}=\sum_{i=1}^{2}{\cal P}(i)/2$,  ${\cal R}=\sum_{i=1}^{2}{\cal R}(i)/2$ and ${\cal F}_1=\sum_{i=1}^{2}{\cal F}_1(i)/2$, respectively. If the competing classes have equal number of observations in the test sample, then $\mathcal R = (n_{11} + n_{22})/n_{00}$, which is the accuracy of the classifier. Note that higher value of precision (respectively, recall) followed by lower value of recall (respectively, precision) and lower $F_1$ score indicates sensitivity to imbalanced data sets.

\begin{table}[t]
	\centering
		\caption{Allocation matrix of a classifier.\label{tab:allocation_matrix}}
	\begin{tabular}{c|c|c|c|c|}
	\multicolumn{1}{c}{} & \multicolumn{1}{c}{} &\multicolumn{2}{c}{Predicted class label}& \multicolumn{1}{c}{} \\ \cline{2-5}
	 & & Class-1 & Class-2 & Total \\ \cline{2-5}
	Actual & Class-1 & $n_{11}$ & $n_{12}$ &$n_{10}$\\ \cline{2-5}
	Class label & Class-2 & $n_{21}$ & $n_{22}$ & $n_{20}$\\ \cline{2-5}
		& Total & $n_{01}$ & $n_{02}$ & $n_{00}$ \\ \cline{2-5}
	\end{tabular}
\end{table}
	
For each example, we consider $1000$ simulation runs, and the average values of ${\cal P}$, ${\cal R}$ and ${\cal F}_1$ over these $1000$ trials are computed for the proposed classifier. These average values are reported in Tables~\ref{tab:binary_normal_location} and \ref{tab:binary_normal_scale} along with their corresponding standard errors. For proper evaluation of the proposed method, we compare its performance with the usual $k$-nearest neighbor classifier (NN), a weighted version of $k$-nearest neighbor classifier (Wt.\ NN) and several data balancing algorithms available in the literature. For Wt.\ NN, we put weight $1/n_i$ on each observation from Class-$i$ ($i=1,2$) to circumvent the effect of data imbalance. Among the data balancing algorithms, we used the ones available in the \texttt{imbalanced-learn} package in \texttt{Python} \citep[][also \url{https://imbalanced-learn.org/stable/}]{imblearn}. In particular, among the undersampling methods, we used condensed nearest neighbors \citep[CondensedNN,][]{hart1968}, one-sided selection \citep[OneSidedSelection,][]{hart1968,kubat1997}, edited nearest neighbors \citep[EditedNN,][]{wilson1972}, repeated edited nearest neighbors \citep[RepeatedEditedNN,][]{tomek1976a}, all $k$-NN \citep[AllKNN,][]{tomek1976a}, Tomek's links \citep[TomekLinks,][]{tomek1976b}, neighborhood cleaning rule \citep[NbdCleaning,][]{laurikkala2001}, near miss \citep[NearMiss,][]{mani2003}, Instance Hardness Threshold \citep[IHT,][]{smith2014},  the random undersampler \citep[RandUnd,][]{dittman2014,mishra2017} and cluster centroids \citep[Centroids,][]{lin2017}. For oversampling, we used linear smote \citep[Smote,][]{chawla2002}, borderline smote \citep[Bsmote,][]{han2005}, adaptive synthetic sampling \citep[ADASYN,][]{he2008}, SVM smote \citep[Ssmote][]{nguyen2011}, smote based on smooth bootstrap \citep[Bootstrap,][]{menardi2014} and $k$-means smote \citep[Ksmote,][]{last2017}. In order to keep the comparison concise, here we only report detailed results for RandUnd, Centroids and IHT among the undersampling methods, and Smote, Bsmote and ADASYN among the oversampling methods, since they are the top three performers in the respective class of methods in our numerical studies (see Figure~\ref{fig:ranks_all} in the Appendix). Some of these methods involve some tuning parameters. For NN and Wt.\ NN, we used $5$ nearest neighbors throughout our numerical studies. For the data balancing algorithms, we used the default values of the tuning parameters. For our proposed method, we used $k_{\max}=5$ in all the numerical studies.

\begin{table}[h!]
	\centering
	\small
	\setlength{\tabcolsep}{0.04in} 
	\caption{Precision ($\mathcal P$), recall ($\mathcal R$) and $F_1$-score ($\mathcal F_1$) of different classifiers in the normal location problem. The results are reported for different proportions of minority samples ($\alpha$) in the training set. The reported numbers are in $\%$. The best result in each example is underlined.\label{tab:binary_normal_location}}
\begin{tabular}{|c|c|ccHccccccc|}\hline
  & $\alpha$ & NN & Wt.\ NN & RandOv & RandUnd & IHT & Centroids & Smote & Bsmote & ADASYN & Proposed \\\hline
  & 0.05 & 72.26 \se{.08} & 69.12 \se{.06} & 68.55 \se{.07} & 72.40 \se{.07} & 73.95 \se{.05} & 72.17 \se{.07} & 68.02 \se{.07} & 70.95 \se{.07} & 67.65 \se{.07} & \best{74.34 \se{.05}} \\
  $ \cal P $ & 0.1 & 72.39 \se{.06} & 69.80 \se{.06} & 69.07 \se{.06} & 72.52 \se{.06} & \best{74.50 \se{.05}} & 72.06 \se{.06} & 69.02 \se{.06} & 70.91 \se{.06} & 68.43 \se{.06} & 74.35 \se{.05} \\
  & 0.2 & 72.42 \se{.05} & 70.45 \se{.05} & 69.51 \se{.05} & 72.55 \se{.06} & \best{74.48 \se{.05}} & 71.89 \se{.06} & 70.14 \se{.06} & 70.73 \se{.06} & 69.42 \se{.05} & 74.14 \se{.05} \\
  & 0.4 & 72.58 \se{.05} & 72.58 \se{.05} & 71.44 \se{.05} & 72.55 \se{.05} & \best{74.00 \se{.05}} & 72.28 \se{.05} & 71.67 \se{.05} & 70.84 \se{.05} & 71.06 \se{.05} & 73.59 \se{.05} \\ [2pt]
  \hline
  & 0.05 & 54.33 \se{.06} & 66.98 \se{.06} & 63.84 \se{.06} & 72.22 \se{.07} & 69.74 \se{.07} & 71.85 \se{.07} & 66.71 \se{.07} & 67.48 \se{.07} & 66.61 \se{.07} & \best{74.15 \se{.05}} \\
  $ \cal R $ & 0.1 & 59.60 \se{.06} & 69.67 \se{.06} & 67.67 \se{.06} & 72.42 \se{.06} & \best{74.29 \se{.05}} & 71.82 \se{.06} & 68.32 \se{.06} & 69.54 \se{.06} & 68.05 \se{.06} & 74.23 \se{.05} \\
  & 0.2 & 66.22 \se{.05} & 69.32 \se{.05} & 69.42 \se{.05} & 72.48 \se{.06} & 72.50 \se{.05} & 71.75 \se{.06} & 69.92 \se{.06} & 70.59 \se{.06} & 69.38 \se{.05} & \best{74.07 \se{.05}} \\
  & 0.4 & 71.94 \se{.05} & 71.94 \se{.05} & 71.41 \se{.05} & 72.51 \se{.05} & 73.52 \se{.05} & 72.24 \se{.05} & 71.62 \se{.05} & 70.69 \se{.05} & 70.94 \se{.05} & \best{73.55 \se{.05}} \\ [2pt]
  \hline
  & 0.05 & 42.69 \se{.11} & 66.02 \se{.07} & 61.36 \se{.07} & 72.16 \se{.07} & 68.31 \se{.09} & 71.74 \se{.07} & 66.08 \se{.07} & 66.06 \se{.08} & 66.11 \se{.07} & \best{74.09 \se{.05}} \\
  ${\cal F}_1$ & 0.1 & 52.80 \se{.09} & 69.62 \se{.06} & 67.06 \se{.06} & 72.39 \se{.06} & \best{74.24 \se{.05}} & 71.74 \se{.06} & 68.02 \se{.06} & 69.03 \se{.06} & 67.88 \se{.06} & 74.20 \se{.05} \\
  & 0.2 & 63.68 \se{.07} & 68.89 \se{.06} & 69.39 \se{.05} & 72.46 \se{.06} & 71.93 \se{.06} & 71.70 \se{.06} & 69.84 \se{.06} & 70.53 \se{.06} & 69.36 \se{.05} & \best{74.05 \se{.05}} \\
  & 0.4 & 71.74 \se{.05} & 71.74 \se{.05} & 71.40 \se{.05} & 72.50 \se{.05} & 73.39 \se{.05} & 72.23 \se{.05} & 71.60 \se{.05} & 70.63 \se{.05} & 70.90 \se{.05} & \best{73.54 \se{.05}} \\ [2pt]
\hline
\end{tabular}%
\end{table}

We begin with a location problem involving bivariate normal distributions $N({\bf 0}, {\bf I})$ and $N((1,1)^{\top},{\bf I})$. We treat the second population as the minority class and generate $\alpha (<1/2)$ proportion of training sample observations from there. The results for different values of $\alpha$ are reported in Table~\ref{tab:binary_normal_location}. We can see that the proposed method outperformed all other competitors in all cases. The differences with other methods are higher for small values of $\alpha$. As $\alpha$ increased (i.e., the training sample became more balanced), all methods tend to perform better and the difference became smaller. Note that in this example, the Bayes (oracle) classifier has ${\cal P}={\cal R}={\cal F}_1= 76.025\%$. The performance of our method is close to that even for small values of $\alpha$. For small  $\alpha$, as expected, NN has low precision for the minority class and low recall for the majority class. Consequently, the $F_1$ score of NN is low. Wt.\ NN and some data balancing algorithms, particularly, Bsmote and IHT, have similar problems.

Next, we consider a scale problem involving bivariate normal distributions $N({\bf 0}, {\bf I})$ and $N({\bf 0},2{\bf I})$. Unlike the location problem, here the results depend on the choice of the minority class. We treat both the populations as the minority in turn, and in each case, report the results in Table~\ref{tab:binary_normal_scale} for $\alpha=0.1,0.2$ and $0.4$. Table~\ref{tab:binary_normal_scale} shows that the proposed method has the best overall performance, in terms of the $F_1$ score, when the minority class has higher variance. Centroids has better performance when the majority class has higher variance, but its performance is very poor in the other scenario. Among the other competing methods, the performance of IHT is comparable to the proposed method.

\begin{table}[h!]
	\centering
	\small
	\setlength{\tabcolsep}{0.04in}
	\caption{Precision ($\mathcal P$), recall ($\mathcal R$) and $F_1$-score ($\mathcal F_1$) of different classifiers in the normal scale problems. The results are reported for different proportions of minority samples ($\alpha$) in the training set. The reported numbers are in $\%$. The best result in each example is underlined.\label{tab:binary_normal_scale}}
\begin{tabular}{|c|c|ccHccccccc|}\hline
  & $\alpha$ & NN & Wt.\ NN & RandOv & RandUnd & IHT & Centroids & Smote & Bsmote & ADASYN & Proposed \\\hline
  \multicolumn{12}{|c|}{Majority class: $N(\mathbf 0,\mathrm I)$, Minority class: $N(\mathbf 0, 2\mathrm I)$} \\ \hline
  & 0.1 & \best{64.34 \se{.16}} & 55.71 \se{.06} & 55.98 \se{.07} & 56.58 \se{.07} & 58.96 \se{.07} & 51.16 \se{.08} & 55.08 \se{.06} & 56.81 \se{.07} & 54.86 \se{.06} & 58.24 \se{.07} \\
  $ \cal P $ & 0.2 & \best{62.17 \se{.08}} & 55.49 \se{.06} & 55.84 \se{.06} & 56.90 \se{.06} & 58.12 \se{.06} & 52.18 \se{.07} & 56.03 \se{.06} & 56.64 \se{.06} & 55.59 \se{.06} & 58.32 \se{.06} \\
  & 0.4 & \best{58.45 \se{.06}} & 58.45 \se{.06} & 56.69 \se{.05} & 57.20 \se{.06} & 58.03 \se{.05} & 55.89 \se{.06} & 56.78 \se{.06} & 56.25 \se{.06} & 55.91 \se{.06} & 57.93 \se{.06} \\ [2pt]
  \hline
  & 0.1 & 51.46 \se{.02} & 55.53 \se{.06} & 54.93 \se{.05} & 56.50 \se{.07} & \best{58.25 \se{.07}} & 51.12 \se{.07} & 54.68 \se{.06} & 55.71 \se{.06} & 54.53 \se{.06} & 58.13 \se{.07} \\
  $ \cal R $ & 0.2 & 54.47 \se{.04} & 55.00 \se{.05} & 55.73 \se{.05} & 56.85 \se{.06} & 58.03 \se{.06} & 52.13 \se{.06} & 55.83 \se{.06} & 56.34 \se{.05} & 55.49 \se{.05} & \best{58.26 \se{.06}} \\
  & 0.4 & 57.48 \se{.05} & 57.48 \se{.05} & 56.65 \se{.05} & 57.17 \se{.06} & \best{58.02 \se{.05}} & 55.80 \se{.06} & 56.72 \se{.05} & 56.23 \se{.05} & 55.88 \se{.06} & 57.88 \se{.05} \\ [2pt]
  \hline
  & 0.1 & 37.36 \se{.05} & 55.19 \se{.06} & 52.84 \se{.06} & 56.38 \se{.07} & 57.42 \se{.07} & 50.88 \se{.07} & 53.75 \se{.06} & 53.85 \se{.07} & 53.76 \se{.06} & \best{57.99 \se{.07}} \\
  ${\cal F}_1$ & 0.2 & 45.87 \se{.06} & 53.97 \se{.05} & 55.52 \se{.06} & 56.78 \se{.06} & 57.92 \se{.06} & 51.90 \se{.06} & 55.45 \se{.06} & 55.85 \se{.06} & 55.30 \se{.06} & \best{58.19 \se{.05}} \\
  & 0.4 & 56.23 \se{.06} & 56.23 \se{.06} & 56.57 \se{.05} & 57.13 \se{.06} & \best{57.99 \se{.05}} & 55.63 \se{.06} & 56.62 \se{.05} & 56.20 \se{.05} & 55.84 \se{.06} & 57.82 \se{.05} \\ [2pt]
  \hline
  \multicolumn{12}{|c|}{Majority class: $N(\mathbf 0, 2\mathrm I)$, Minority class: $N(\mathbf 0,\mathrm I)$} \\ \hline
  & 0.1 & 55.20 \se{.21} & 54.99 \se{.06} & 54.32 \se{.06} & 56.49 \se{.07} & 58.76 \se{.07} & \best{61.25 \se{.06}} & 55.05 \se{.06} & 55.60 \se{.06} & 55.00 \se{.06} & 58.71 \se{.06} \\
  $ \cal P $ & 0.2 & 55.52 \se{.09} & 57.03 \se{.06} & 55.22 \se{.05} & 57.13 \se{.06} & 60.84 \se{.06} & \best{61.13 \se{.05}} & 55.67 \se{.05} & 56.06 \se{.06} & 55.57 \se{.06} & 58.84 \se{.05} \\
  & 0.4 & 56.35 \se{.06} & 56.35 \se{.06} & 56.34 \se{.05} & 57.16 \se{.06} & 58.58 \se{.06} & \best{59.25 \se{.05}} & 56.49 \se{.05} & 56.37 \se{.05} & 56.39 \se{.05} & 58.10 \se{.05} \\ [2pt]
  \hline
  & 0.1 & 50.35 \se{.01} & 54.90 \se{.06} & 53.60 \se{.05} & 56.42 \se{.07} & 58.68 \se{.07} & \best{60.06 \se{.05}} & 54.65 \se{.06} & 54.72 \se{.05} & 54.66 \se{.06} & 58.43 \se{.06} \\
  $ \cal R $ & 0.2 & 51.86 \se{.03} & 56.20 \se{.05} & 55.16 \se{.05} & 57.08 \se{.06} & 59.67 \se{.05} & \best{60.26 \se{.05}} & 55.52 \se{.05} & 55.85 \se{.05} & 55.49 \se{.06} & 58.65 \se{.05} \\
  & 0.4 & 55.91 \se{.05} & 55.91 \se{.05} & 56.33 \se{.05} & 57.12 \se{.05} & 58.43 \se{.05} & \best{59.05 \se{.05}} & 56.47 \se{.05} & 56.35 \se{.05} & 56.32 \se{.05} & 58.02 \se{.05} \\ [2pt]
  \hline
  & 0.1 & 35.23 \se{.03} & 54.69 \se{.06} & 51.55 \se{.06} & 56.29 \se{.07} & 58.59 \se{.07} & \best{58.99 \se{.06}} & 53.72 \se{.06} & 52.85 \se{.06} & 53.87 \se{.06} & 58.09 \se{.06} \\
  ${\cal F}_1$ & 0.2 & 42.23 \se{.05} & 54.85 \se{.06} & 55.01 \se{.05} & 57.01 \se{.06} & 58.58 \se{.06} & \best{59.47 \se{.05}} & 55.22 \se{.05} & 55.45 \se{.06} & 55.34 \se{.06} & 58.43 \se{.05} \\
  & 0.4 & 55.13 \se{.05} & 55.13 \se{.05} & 56.31 \se{.05} & 57.08 \se{.05} & 58.24 \se{.05} & \best{58.84 \se{.05}} & 56.45 \se{.05} & 56.33 \se{.05} & 56.21 \se{.05} & 57.93 \se{.05} \\ [2pt]
	\hline
\end{tabular}%
\end{table}%

We carried out our experiments for other distributions as well, but the results were more or less similar. Barring a  few cases, the proposed method outperformed all its competitors considered here. Therefore, to save space, we do not report them. In the case of balanced training sample, while Wt.\ NN and all data balancing algorithms coincide with the usual nearest neighbor classifier, the proposed method leads to a slightly different classifier. In the normal location problem with balanced samples, while NN has average precision, recall and $F_1$ score of $75.11\%$, $75.02\%$ and $75\%$, respectively, those for the proposed method are $75.28\%$, $75.19\%$ and $75.16\%$ in the respective cases. The performances in the balanced normal scale problem are also comparable -- the average precision, recall and $F_1$ score for NN are $60.49\%$, $60.12\%$, $59.78\%$, respectively, compared to a slightly better performance of $62.67\%$, $61.18\%$ and $60.02\%$ by the proposed method.

\subsection{Analysis of Benchmark Datasets}
\label{sec:binary_realdata}

Here, we analyze 14 benchmark data sets for further evaluation of the proposed method. Breast Cancer, Haberman, Predictive Maintenance and Pima Indian data sets are taken from the UCI Machine Learning Repository (\url{https://archive.ics.uci.edu/ml/datasets.php}). The rest of the data sets are available at Kaggle (\url{https://www.kaggle.com/datasets}).  For Asteroids and Loan data sets, we remove all redundant variables before carrying out our experiment. The Fetal Health data set has observations from three classes: `Normal', `Suspect' and `Pathological',. However, the distinction between the last two classes is not very clear. So, we consider it as a two-class problem, where `Suspect' and `Pathological' is considered as one class, and `Normal' as the other class. In most of these data sets, the measurement variables are not of comparable units and scales. So, we standardize each of the measurement variables and work with the standardized data sets. We divide each data set randomly into two groups to form the training and test sets.  We first divide the minority class observations in two groups containing nearly 75\% and 25\% observations. Equal number of observations from the majority class are added to the smaller group to form the test set, while the rest of the majority class observations are added to the larger group to form the training set. For each data set, this random partitioning is done 1000 times, and the average overall performance of different classifiers (computed over these 1000 trials) are reported in Table~\ref{tab:binary_benchmark}.

\begin{table}[h!]
	\centering
	\scriptsize
	\setlength{\tabcolsep}{0.04in}
	\caption{Precision ($\mathcal P$), recall ($\mathcal R$) and $F_1$-score ($\mathcal F_1$) of different classifiers on benchmark data sets involving two classes. Brief description of the datasets are given before providing the results. The reported numbers are in $\%$. The best result in each example is underlined.\label{tab:binary_benchmark}}
	\begin{tabular}{|c|cccccccccccccc|}
\multicolumn{15}{c}{Dataset description} \\ \hline
Dataset & Astro- & Breast & Fetal & Haber- & Loan & NASA & Paris & Pima & Pred- & Rice & Room & Stroke & Water & Wine \\
 & Seismo- & Cancer & Health & man & & & & & Maint & & & & &  \\
 & logy & & & & & & & & & & & & &  \\ \hline 
Dimension & 3 & 30 & 21 & 3 & 11 & 20 & 17 & 8 & 8 & 10 & 5 & 13 & 9 & 11 \\ \hline
Training & 641 & 304 & 1538 & 205 & 4400 & 3744 & 8419 & 433 & 9577 & 7935 & 1450 & 4799 & 1679 & 669 \\
sample size & 216 & 159 & 354 & 61 & 360 & 567 & 949 & 201 & 255 & 6150 & 729 & 187 & 959 & 558 \\ \hline
Imbalance ratio & 0.252 & 0.343 & 0.187 & 0.229 & 0.076 & 0.131 & 0.101 & 0.317 & 0.026 & 0.437 & 0.335 & 0.038 & 0.364 & 0.455 \\ \hline
Test sample &  & & & & & & & & & & & & & \\
size (per class) & 72 & 53 & 117 & 20 & 120 &188 & 316 & 67 & 84 & 2050 & 243 & 62 & 319 & 186 \\ \hline
\end{tabular}

\medskip

\begin{tabular}{|c|HHc|ccHccccccc|}
 \multicolumn{14}{c}{Results} \\ \hline
Dataset  & Train & Test & & NN & Wt.\ NN & RandOv & RandUnd & IHT & Centroids & Smote & Bsmote & ADASYN & Proposed \\ \hline
  Astro- & 641 & 72 & $ \cal P $ & 95.51 \se{.05} & \best{95.94 \se{.05}} & 95.85 \se{.05} & 95.42 \se{.05} & 94.99 \se{.05} & 95.83 \se{.05} & 95.88 \se{.05} & 95.36 \se{.05} & 95.32 \se{.05} & 95.62 \se{.05} \\
  Seismo- & 216 & 72 & $ \cal R $ & 95.38 \se{.05} & \best{95.89 \se{.05}} & 95.80 \se{.05} & 95.33 \se{.05} & 94.75 \se{.06} & 95.74 \se{.05} & 95.82 \se{.05} & 95.27 \se{.05} & 95.21 \se{.05} & 95.53 \se{.05} \\
  logy & & & ${\cal F}_1$ & 95.38 \se{.05} & \best{95.89 \se{.05}} & 95.79 \se{.05} & 95.33 \se{.05} & 94.75 \se{.06} & 95.73 \se{.05} & 95.82 \se{.05} & 95.27 \se{.05} & 95.20 \se{.05} & 95.53 \se{.05} \\ [2pt] \hline
  Breast & 304 & 53 & $ \cal P $ & 95.92 \se{.05} & 96.08 \se{.06} & 95.89 \se{.06} & 95.77 \se{.06} & 94.86 \se{.06} & 95.99 \se{.05} & 96.08 \se{.05} & 94.95 \se{.06} & 94.72 \se{.06} & \best{96.38 \se{.05}} \\
  Cancer & 159 & 53 & $ \cal R $ & 95.64 \se{.06} & 96.00 \se{.06} & 95.79 \se{.06} & 95.67 \se{.06} & 94.71 \se{.06} & 95.87 \se{.06} & 95.99 \se{.06} & 94.82 \se{.06} & 94.53 \se{.06} & \best{96.28 \se{.05}} \\
   & & & ${\cal F}_1$ & 95.63 \se{.06} & 96.00 \se{.06} & 95.79 \se{.06} & 95.67 \se{.06} & 94.71 \se{.06} & 95.86 \se{.06} & 95.98 \se{.06} & 94.82 \se{.06} & 94.53 \se{.06} & \best{96.28 \se{.05}} \\ [2pt] \hline
  Fetal & 1538 & 117 & $ \cal P $ & 87.15 \se{.05} & 91.02 \se{.06} & 90.28 \se{.06} & 89.21 \se{.06} & 89.70 \se{.06} & 89.01 \se{.06} & 90.69 \se{.06} & 90.89 \se{.06} & \best{91.36 \se{.06}} & 91.22 \se{.06} \\
  Health & 354 & 117 & $ \cal R $ & 84.63 \se{.07} & 90.68 \se{.06} & 90.19 \se{.06} & 89.11 \se{.06} & 89.33 \se{.06} & 88.91 \se{.06} & 90.61 \se{.06} & 90.82 \se{.06} & \best{91.27 \se{.06}} & 91.08 \se{.06} \\
   & & & ${\cal F}_1$ & 84.35 \se{.07} & 90.66 \se{.06} & 90.18 \se{.06} & 89.10 \se{.06} & 89.30 \se{.06} & 88.90 \se{.06} & 90.60 \se{.06} & 90.81 \se{.06} & \best{91.26 \se{.06}} & 91.07 \se{.06} \\ [2pt] \hline
  Haber- & 205 & 20 & $ \cal P $ & 60.01 \se{.36} & 59.12 \se{.24} & 61.23 \se{.23} & 58.59 \se{.23} & \best{62.37 \se{.24}} & 52.28 \se{.23} & 60.47 \se{.24} & 60.65 \se{.24} & 60.46 \se{.23} & 62.02 \se{.23} \\
  man & 61 & 20 & $ \cal R $ & 54.25 \se{.15} & 58.02 \se{.21} & 60.79 \se{.22} & 58.27 \se{.23} & \best{61.97 \se{.23}} & 52.20 \se{.22} & 59.93 \se{.23} & 60.05 \se{.23} & 60.09 \se{.23} & 61.62 \se{.22} \\
  & & & ${\cal F}_1$ & 46.94 \se{.19} & 56.79 \se{.22} & 60.37 \se{.23} & 57.89 \se{.23} & \best{61.64 \se{.24}} & 51.78 \se{.22} & 59.42 \se{.23} & 59.49 \se{.23} & 59.71 \se{.23} & 61.28 \se{.23} \\ [2pt] \hline
  Loan & 4400 & 120 & $ \cal P $ & 84.82 \se{.04} & 91.23 \se{.05} & 89.98 \se{.05} & 90.40 \se{.06} & 88.53 \se{.05} & 91.18 \se{.06} & 91.32 \se{.05} & 91.08 \se{.05} & 91.58 \se{.05} & \best{92.61 \se{.05}} \\
  & 360 & 120 & $ \cal R $ & 78.55 \se{.07} & 90.89 \se{.05} & 88.90 \se{.06} & 90.33 \se{.06} & 86.72 \se{.06} & 91.09 \se{.06} & 90.68 \se{.05} & 90.36 \se{.06} & 91.11 \se{.05} & \best{92.54 \se{.05}} \\
  & & & ${\cal F}_1$ & 77.52 \se{.07} & 90.88 \se{.05} & 88.82 \se{.06} & 90.32 \se{.06} & 86.55 \se{.06} & 91.09 \se{.06} & 90.64 \se{.06} & 90.31 \se{.06} & 91.08 \se{.05} & \best{92.54 \se{.05}} \\ [2pt] \hline
  NASA & 3744 & 188 & $ \cal P $ & 80.95 \se{.05} & 87.53 \se{.05} & 86.50 \se{.05} & 86.45 \se{.05} & 85.37 \se{.06} & 87.25 \se{.05} & 86.87 \se{.05} & 86.78 \se{.05} & 87.08 \se{.05} & \best{88.57 \se{.04}} \\
  & 567 & 188 & $ \cal R $ & 74.73 \se{.06} & 87.09 \se{.05} & 86.42 \se{.05} & 85.63 \se{.06} & 85.05 \se{.06} & 86.97 \se{.05} & 86.83 \se{.05} & 86.73 \se{.05} & 87.02 \se{.05} & \best{87.41 \se{.05}} \\
   & & & ${\cal F}_1$ & 73.38 \se{.06} & 87.05 \se{.05} & 86.42 \se{.05} & 85.55 \se{.06} & 85.02 \se{.06} & 86.95 \se{.05} & 86.82 \se{.05} & 86.73 \se{.05} & 87.01 \se{.05} & \best{87.31 \se{.05}} \\ [2pt] \hline
  Paris & 8419 & 316 & $ \cal P $ & 96.74 \se{.02} & 96.65 \se{.02} & 98.09 \se{.02} & 93.19 \se{.03} & \best{99.19 \se{.01}} & 97.91 \se{.02} & 97.81 \se{.02} & 98.11 \se{.02} & 97.90 \se{.02} & 94.61 \se{.02} \\
  & 949 & 316 & $ \cal R $ & 96.51 \se{.02} & 96.41 \se{.02} & 98.05 \se{.02} & 92.12 \se{.03} & \best{99.18 \se{.01}} & 97.82 \se{.02} & 97.72 \se{.02} & 98.09 \se{.02} & 97.81 \se{.02} & 93.95 \se{.03} \\
  & & & ${\cal F}_1$ & 96.51 \se{.02} & 96.40 \se{.02} & 98.05 \se{.02} & 92.07 \se{.03} & \best{99.18 \se{.01}} & 97.82 \se{.02} & 97.72 \se{.02} & 98.09 \se{.02} & 97.81 \se{.02} & 93.93 \se{.03} \\ [2pt] \hline
  Pima & 433 & 67 & $ \cal P $ & 70.90 \se{.11} & 71.98 \se{.11} & 71.08 \se{.11} & 71.10 \se{.11} & \best{74.00 \se{.11}} & 70.04 \se{.12} & 71.82 \se{.11} & 71.40 \se{.11} & 71.13 \se{.11} & 73.38 \se{.11} \\
  & 201 & 67 & $ \cal R $ & 68.03 \se{.1} & 71.84 \se{.11} & 70.90 \se{.11} & 70.94 \se{.11} & 73.09 \se{.11} & 69.76 \se{.11} & 71.66 \se{.11} & 71.24 \se{.11} & 70.96 \se{.11} & \best{73.23 \se{.11}} \\
  & & & ${\cal F}_1$ & 66.88 \se{.11} & 71.79 \se{.11} & 70.84 \se{.11} & 70.89 \se{.11} & 72.82 \se{.12} & 69.66 \se{.11} & 71.61 \se{.11} & 71.18 \se{.11} & 70.90 \se{.11} & \best{73.18 \se{.11}} \\ [2pt] \hline
  Pred- & 9577 & 84 & $ \cal P $ & 78.44 \se{.04} & 85.05 \se{.07} & 83.46 \se{.07} & 86.06 \se{.08} & 83.48 \se{.06} & 86.74 \se{.08} & 85.15 \se{.08} & 85.46 \se{.07} & 85.14 \se{.08} & \best{87.47 \se{.08}} \\
  Maint & 255 & 84 & $ \cal R $ & 62.89 \se{.06} & 82.56 \se{.08} & 78.92 \se{.08} & 85.93 \se{.08} & 77.97 \se{.08} & 86.62 \se{.08} & 83.52 \se{.08} & 82.73 \se{.08} & 83.72 \se{.08} & \best{87.36 \se{.08}} \\
  & & & ${\cal F}_1$ & 56.94 \se{.1} & 82.23 \se{.09} & 78.16 \se{.09} & 85.92 \se{.08} & 76.99 \se{.09} & 86.61 \se{.08} & 83.32 \se{.09} & 82.37 \se{.08} & 83.55 \se{.08} & \best{87.35 \se{.08}} \\ [2pt] \hline
  Rice & 7935 & 2050 & $ \cal P $ & 98.82 \se{.01} & 98.82 \se{.01} & 98.75 \se{.01} & 98.82 \se{.01} & 98.27 \se{.01} & 98.82 \se{.01} & 98.77 \se{.01} & 98.16 \se{.01} & 97.60 \se{.01} & \best{98.85 \se{.01}} \\
  & 6150 & 2050 & $ \cal R $ & 98.81 \se{.01} & 98.81 \se{.01} & 98.75 \se{.01} & 98.81 \se{.01} & 98.25 \se{.01} & 98.81 \se{.01} & 98.77 \se{.01} & 98.14 \se{.01} & 97.55 \se{.01} & \best{98.85 \se{.01}} \\
  & & & ${\cal F}_1$ & 98.81 \se{.01} & 98.81 \se{.01} & 98.75 \se{.01} & 98.81 \se{.01} & 98.25 \se{.01} & 98.81 \se{.01} & 98.77 \se{.01} & 98.14 \se{.01} & 97.55 \se{.01} & \best{98.85 \se{.01}} \\ [2pt] \hline
  Room & 1450 & 243 & $ \cal P $ & 98.58 \se{.02} & 98.76 \se{.01} & 98.79 \se{.02} & 98.54 \se{.02} & 98.47 \se{.02} & 98.62 \se{.02} & 98.78 \se{.02} & 98.74 \se{.02} & \best{98.85 \se{.01}} & 98.69 \se{.02} \\
  & 729 & 243 & $ \cal R $ & 98.57 \se{.02} & 98.75 \se{.02} & 98.78 \se{.02} & 98.51 \se{.02} & 98.41 \se{.02} & 98.61 \se{.02} & 98.77 \se{.02} & 98.73 \se{.02} & \best{98.83 \se{.01}} & 98.67 \se{.02} \\
  & & & ${\cal F}_1$ & 98.57 \se{.02} & 98.75 \se{.02} & 98.78 \se{.02} & 98.51 \se{.02} & 98.41 \se{.02} & 98.61 \se{.02} & 98.77 \se{.02} & 98.73 \se{.02} & \best{98.83 \se{.01}} & 98.67 \se{.02} \\ [2pt] \hline
  Stroke & 4799 & 62 & $ \cal P $ & 78.41 \se{.06} & 90.18 \se{.08} & 87.30 \se{.08} & 88.32 \se{.08} & 82.06 \se{.08} & 88.47 \se{.09} & 90.06 \se{.08} & 89.06 \se{.08} & \best{90.37 \se{.08}} & 89.93 \se{.07} \\
  & 187 & 62 & $ \cal R $ & 63.84 \se{.08} & 89.70 \se{.08} & 85.46 \se{.09} & 87.26 \se{.09} & 76.09 \se{.1} & 87.92 \se{.09} & 89.45 \se{.08} & 88.15 \se{.09} & \best{89.86 \se{.08}} & 88.69 \se{.09} \\
  & & & ${\cal F}_1$ & 58.44 \se{.11} & 89.67 \se{.09} & 85.27 \se{.09} & 87.16 \se{.09} & 74.86 \se{.12} & 87.87 \se{.09} & 89.41 \se{.08} & 88.07 \se{.09} & \best{89.82 \se{.08}} & 88.58 \se{.09} \\ [2pt] \hline
  Water & 1679 & 319 & $ \cal P $ & 59.97 \se{.06} & 58.61 \se{.06} & 58.20 \se{.06} & 58.56 \se{.06} & \best{59.98 \se{.06}} & 58.51 \se{.06} & 58.46 \se{.06} & 58.27 \se{.06} & 58.39 \se{.06} & 59.89 \se{.06} \\
  & 959 & 319 & $ \cal R $ & 57.77 \se{.05} & 58.45 \se{.06} & 58.11 \se{.06} & 58.52 \se{.06} & \best{59.90 \se{.06}} & 58.39 \se{.06} & 58.42 \se{.06} & 58.25 \se{.06} & 58.37 \se{.06} & 59.86 \se{.06} \\
  & & & ${\cal F}_1$ & 55.31 \se{.05} & 58.26 \se{.06} & 58.00 \se{.06} & 58.47 \se{.06} & 59.83 \se{.06} & 58.25 \se{.06} & 58.38 \se{.06} & 58.23 \se{.06} & 58.35 \se{.06} & \best{59.83 \se{.06}} \\ [2pt] \hline
  Wine & 669 & 186 & $ \cal P $ & 72.22 \se{.07} & 72.22 \se{.07} & 71.98 \se{.07} & 72.07 \se{.07} & 73.18 \se{.06} & 72.28 \se{.07} & 72.16 \se{.07} & 71.83 \se{.07} & 71.06 \se{.07} & \best{74.59 \se{.06}} \\
  & 558 & 186 & $ \cal R $ & 71.75 \se{.07} & 71.75 \se{.07} & 71.84 \se{.07} & 71.95 \se{.07} & 73.10 \se{.06} & 72.19 \se{.06} & 72.06 \se{.07} & 71.77 \se{.07} & 70.88 \se{.06} & \best{74.50 \se{.06}} \\
  & & & ${\cal F}_1$ & 71.60 \se{.07} & 71.60 \se{.07} & 71.80 \se{.07} & 71.92 \se{.07} & 73.08 \se{.06} & 72.16 \se{.07} & 72.03 \se{.07} & 71.75 \se{.07} & 70.81 \se{.07} & \best{74.47 \se{.06}} \\ [2pt] \hline
\end{tabular}%

\vspace{0.025in}
\footnotesize{~$^\ast$ Some redundant features were removed before our analysis.}%
\end{table}

Table~\ref{tab:binary_benchmark} clearly shows that the performance of the proposed method is highly satisfactory. In 6 out of these 12 data sets (Breast Cancer, Loan, NASA, Predictive Maintenance, Rice and Wine), our method has the best precision, recall and $F_1$ score.  It has the best $F_1$ score on the Pima Indian and Water data sets as well. On Fetal Health and Haberman data, the proposed method has the second best performance, while it has the third best performance in Stroke data. On all these data sets, the performance of the proposed method is among the best four methods.

To compare the overall performance of different methods in a comprehensive way, following the ideas of \cite{chaudhuri2008,ghosh2012}, we introduce the notion of efficiency of a classifier. If ${\cal P}_1,{\cal P}_2,\ldots,{\cal P}_T$ are the precision of $T$ classifiers on a data set, the precision-efficiency of the $i$-th classifier is defined as ${\cal P}_i/\max_{1\le j\le T} {\cal P}_j$ for $i=1,2,\ldots,T$. So, for a data set, the best classifier has efficiency $1$, whereas a small efficiency indicates poor performance by a classifier compared to its competitors. For each data set, we compute these values for all the classifiers, and they are graphically represented by boxplots in Figure~\ref{fig:binary_prec_efficiency}. This figure clearly shows that the overall performance of the proposed method was much better than its competitors. Similarly, one can compute the recall-efficiency and the $F_1$-efficiency of different classifiers and construct the corresponding boxplots. But those boxplots were similar to the ones given in Figure~\ref{fig:binary_prec_efficiency}, so we do not report them here.

\begin{figure}[h!]
	\centering
	\includegraphics[height=2.3in]{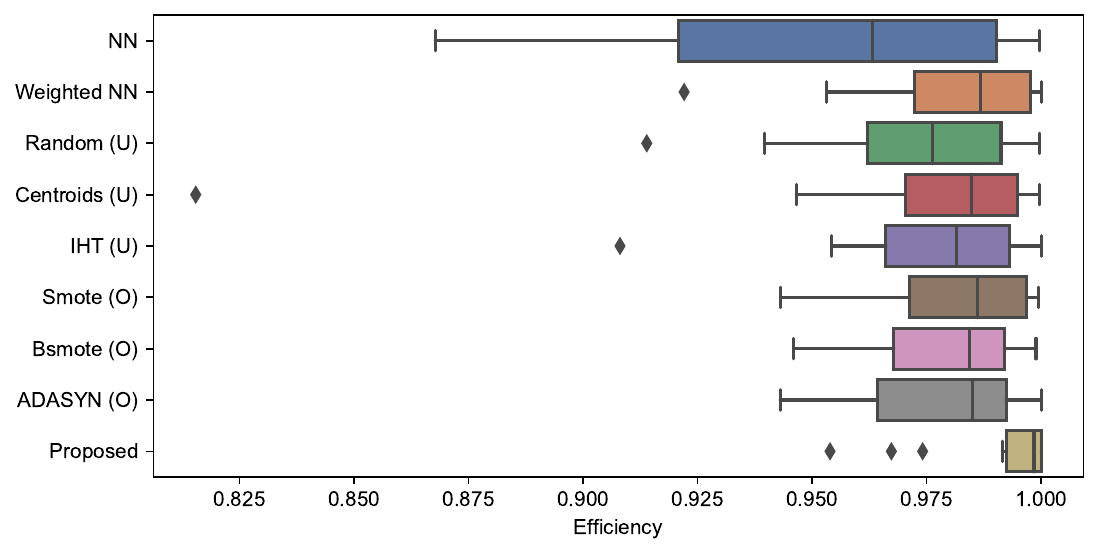}
	\caption{Boxplots of precision-efficiency of different classifiers on two-class benchmark data sets. The undersampling methods are indicated by (U) and the oversampling methods are indicated by (O)\label{fig:binary_prec_efficiency}}
\end{figure}

\section{Classification involving multiple classes}
\label{sec:multi-class}

The proposed method can be generalized for $J$-class problems with $J>2$. One simple way of generalization is to adopt the one-vs-one (OvO) approach, where we perform $\binom{J}{2}$ binary classifications taking one pair of classes at a time and then use majority voting. On some rare occasions, this may lead to ties, where more than one classes have the maximum number of votes. One can break these ties by considering a classification problem involving those classes only. If the problem is still not resolved, then the class having the maximum evidence in terms of the $p$-value (as described in the case of binary classification) can be chosen. However, the OvO approach has high computing cost since we need to perform $\binom{J}{2}$ binary classifications. We can adopt a slight modification of OvO to reduce the computing cost. First, we arrange the classes in descending order of the sample sizes (ties are resolved arbitrarily). Next, we perform $J-1$ binary classifications treating the $J$-th class as the minority class and one of the rest as the majority class. For an observation ${\bf x}$, let $S_{\bf x}(J)$ be the collection of all classes winning over the minority class and $J_0=|S_{\bf x}(J)|$ be the cardinality of $S_{\bf x}(J)$. If $S_{\bf x}(J)$ is empty (i.e., $J_0=0$), we classify ${\bf x}$ to Class-$J$. If $J_0=1$, we assign ${\bf x}$ to the single member of $S_{\bf x}(J)$. Otherwise, we repeat the procedure for a $J_0$-class problem, involving the classes in $S_{\bf x}(J)$ only. In our numerical studies, there was no visible difference between the performance of these two methods, and here we report the results for the second one (referred to as the OvO$+$ method). A consistency result similar to Theorem~\ref{thm:binary_consistency} can be proved for this method, which is stated in Theorem~\ref{thm:consistency_multiclass}.

Instead of one-vs-one, one can also use the one-vs-rest (OvR) method, where each time we consider a binary classification problem between one class and a combined class having observations from the rest of the classes. After $J$ binary classifications, majority voting is used to arrive at a final decision. Here also, we may have ties in some cases, and the method based on the maximum $p$-value type evidence can be used to resolve these ties. Another option is to consider a classification problem among the tied classes and repeat the procedure until we get a winner. It can be shown that as the sample size increases, the probability of getting a unique winner by this method converges to $1$ and the resulting classifier is consistent. This result is stated in Theorem~\ref{thm:consistency_multiclass}. However, in practice, it may be possible to have some rare cases, where the ties cannot be resolved in this way. The method based on maximum $p$-value type evidence can be used in such rare cases. In our numerical studies, both of these methods led to almost similar results in all data sets. So, here we report the results for the second method (referred to as the OvR$+$ method) only.

\begin{theorem}\label{thm:consistency_multiclass}
Suppose that the $J$ competing classes have continuous probability density functions $f_1,\ldots,f_J$, respectively. Assume that for each $i=1,\ldots,J$, $n_i/n \to \pi_i \in (0,1)$ as $n \to \infty$ ($\sum_{i=1}^{J} \pi_i=1$). Also assume that for each binary classification (both in one-vs-one and one-vs-rest methods), $k_{\max}$, the maximum number of neighbors from the minority class diverges to infinity in such a way that $k_{\max}/n \rightarrow 0$ as $n \to \infty$. Then the proposed methods OvO$+$ and OvR$+$ classify an observation ${\bf x}$ to the class $i_0=\arg\max_{1 \le i\le J} f_i({\bf x})$ with probability tending to $1$ as $n$ tends to infinity.
\end{theorem}

\begin{proof} ($i$) \underline{OvO$+$ method}: At the first step, we consider $J-1$ binary classifications, taking the $J$-th class and one of the other classes at a time. Let $S=\{i: f_i({\bf x})>f_J({\bf x})\}$. Note that for each of these binary classification problems, the conditions of Theorem~\ref{thm:binary_consistency} are satisfied. Since $J$ is finite, following the proof of Theorem~\ref{thm:binary_consistency}, one can show that $P(S_{\bf x}(J)=S) \rightarrow 1$ as $n \rightarrow \infty$. So, for $J_0=|S_{\bf x}(J)|=0$ or $1$, the proof follows immediately. For $J_0>1$, we have $P(i_0 \in S_{\bf x}(J)) \rightarrow 1$ and the result is obtained by repeating the same argument for the classification problem involving $J_0$ ($J_0 \le J-1$) classes in ${S_{\bf x}(J)}$.

\noindent
($ii$) \underline{OvR$+$ method}: Consider the classification problem between Class-$i$ ($i=1,\ldots,J$) and the combined class containing the rest. Note that this combined class, being a mixture of $J-1$ classes, has density $f_{-i}({\bf x})=\sum_{j\neq i}\pi_jf_j({\bf x})/\sum_{j \neq i}\pi_j$. Now, from Theorem~\ref{thm:binary_consistency}, it is easy to see that for an observation ${\bf x}$, our binary classifier will prefer Class-$i$ with probability tending to $1$ if and only if  $f_i({\bf x})>f_{-i}({\bf x}) \Leftrightarrow f_i({\bf x})>\sum_{j=1}^{J}\pi_jf_j({\bf x})=f_0({\bf x})$, say. Now, define $S^{\ast}=\{j: f_j({\bf x})>f_0({\bf x})\}$. So, if $i \in S^{\ast}$ (respectively, $i \notin S^{\ast}$), Class-$i$ wins over (respectively, loses to) the rest with probability tending to $1$ as $n$ tends to infinity. From the definition of $i_0$, it is clear that $i_0 \in S^{\ast}$. Therefore, if $J_*=|S^{\ast}|=1$, the result follows immediately. If $J_*>1$ , we need to consider a classification problem involving $J_*$ classes in $S^{\ast}$ (note that $J_*<J$). Since $J$ is finite, repeated application of the same argument leads to the proof. 
\end{proof}

\begin{table}[h!]
\centering
\scriptsize
\setlength{\tabcolsep}{0.05in}
\caption{Precision ($\mathcal P$), recall ($\mathcal R$) and $F_1$-score ($\mathcal F_1$) of different classifiers on benchmark data sets involving more than two classes. Brief description of the datasets are given before providing the results. The reported numbers are in $\%$. The best result in each example is underlined. \label{tab:benchmark_multiclass}}
\begin{tabular}{|c|ccccccccccc|}
\multicolumn{12}{c}{Dataset description} \\ \hline
Dataset & Drugs & Dry & EColi & Fetal & Letter & Sat & Shuttle & SVM & SVM & Vehicle & Wine  \\
 & & Bean & & Health & & Image & & Guide 2 & Guide 4 & & Recog \\ \hline
Dimension & 5 & 16 & 6 & 21 & 16 & 36 & 9 & 20 & 10 & 18 & 13  \\ \hline
No. of classes & 5 & 7 & 5 & 3 & 26 & 6 & 3 & 3 & 6 & 4 & 3 \\ \hline 
Training & 87 & 3416 & 138 & 1611 & 473 472 471 468 & 969 & 10022 & 208 & 99 & 169 & 60 \\
sample & 50 & 2506 & 72 & 251 & 465 464 463 460 & 935 & 1851 & 104 & 96 & 168 & 46 \\ 
 size & 19 & 1897 & 47 & 132 & 460 449 447 442 & 858 & 561 & 40 & 90 & 163 & 36  \\ 
 & 12 & 1798 & 30 & & 441 440 438 434 & 376 & & & 79 & 150 &  \\
 & 12 & 1500 & 15 & & 433 433 431 431 & 367 & & & 66 & & \\
 &  & 1192 & & & 429 426 420 413 & 312 & & & 62 & & \\
 &  & 392 & & &  409 404 & & & & & & \\ \hline
Test sample &  & & & & & & & & & & \\
size (per class) & 4 & 130 & 5 & 44 & 134 & 103 & 187 & 13 & 20 & 49 & 11 \\ \hline
\end{tabular}

\medskip
\setlength{\tabcolsep}{0.03in}
\begin{tabular}{|c|HHc|ccHcccccccc|}
 \multicolumn{15}{c}{Results} \\ \hline
  Dataset & Train & Test & & NN & Wt.\ NN & RandOv & RandUnd & IHT & Centroids & Smote & Bsmote & ADASYN & \multicolumn{2}{c|}{Proposed} \\
  & & & & & & & & & & & & & OvO+ & OvR+ \\ \hline
  Drugs & 87 50 19 12 12 & 4 & $ \cal P $ & 81.49 \se{.28} & 81.69 \se{.29} & 90.70 \se{.20} & 65.89 \se{.41} & 66.18 \se{.41} & 65.41 \se{.41} & \best{90.72 \se{.20}} & 90.63 \se{.21} & 82.94 \se{.27} & 79.40 \se{.32} & 82.70 \se{.28} \\
  & & & $ \cal R $ & 77.94 \se{.28} & 77.87 \se{.27} & 88.53 \se{.21} & 60.28 \se{.33} & 60.49 \se{.34} & 58.14 \se{.33} & \best{88.66 \se{.21}} & 88.47 \se{.22} & 79.34 \se{.28} & 73.12 \se{.30} & 78.32 \se{.30} \\
  & & & ${\cal F}_1$ & 77.06 \se{.29} & 76.54 \se{.30} & 87.65 \se{.24} & 57.82 \se{.35} & 58.70 \se{.36} & 55.90 \se{.36} & \best{87.78 \se{.24}} & 87.58 \se{.24} & 78.49 \se{.30} & 71.50 \se{.33} & 77.33 \se{.32} \\ [2pt] \hline
  Dry & 3416 2506 1897 1798 1500 1192 392 & 130 & $ \cal P $ & 93.55 \se{.02} & 93.49 \se{.03} & 92.78 \se{.03} & 92.84 \se{.03} & 92.61 \se{.03} & 93.34 \se{.03} & 93.20 \se{.03} & 91.95 \se{.03} & 93.51 \se{.02} & 93.71 \se{.03} & \best{93.71 \se{.02}} \\
  Bean & & & $ \cal R $ & 93.42 \se{.03} & 93.41 \se{.03} & 92.71 \se{.03} & 92.72 \se{.03} & 92.22 \se{.03} & 93.24 \se{.03} & 93.13 \se{.03} & 91.86 \se{.03} & 93.38 \se{.02} & \best{93.61 \se{.03}} & 93.59 \se{.03} \\
  & & & ${\cal F}_1$ & 93.44 \se{.03} & 93.42 \se{.03} & 92.72 \se{.03} & 92.74 \se{.03} & 92.29 \se{.03} & 93.25 \se{.03} & 93.14 \se{.03} & 91.86 \se{.03} & 93.40 \se{.02} & \best{93.62 \se{.03}} & 93.61 \se{.03} \\ [2pt] \hline
  EColi & 138 72 47 30 15 & 5 & $ \cal P $ & 83.21 \se{.22} & 83.64 \se{.23} & 80.03 \se{.24} & 84.25 \se{.24} & 85.90 \se{.22} & 85.12 \se{.22} & 81.83 \se{.23} & 78.17 \se{.24} & 82.22 \se{.23} & 85.63 \se{.22} & \best{86.46 \se{.21}} \\
  & & & $ \cal R $ & 80.39 \se{.21} & 81.99 \se{.21} & 78.15 \se{.24} & 82.52 \se{.23} & 83.94 \se{.21} & 83.28 \se{.22} & 80.05 \se{.23} & 76.41 \se{.23} & 79.69 \se{.21} & 83.90 \se{.21} & \best{84.85 \se{.21}} \\
  & & & ${\cal F}_1$ & 79.37 \se{.23} & 81.36 \se{.23} & 77.58 \se{.24} & 81.80 \se{.24} & 83.24 \se{.23} & 82.56 \se{.23} & 79.51 \se{.23} & 75.73 \se{.23} & 78.53 \se{.23} & 83.30 \se{.22} & \best{84.41 \se{.22}} \\ [2pt] \hline
  Fetal & 1611 251 132 & 44 & $ \cal P $ & 80.37 \se{.09} & 86.13 \se{.09} & 86.05 \se{.09} & 83.99 \se{.09} & 82.54 \se{.09} & 84.19 \se{.09} & 86.74 \se{.09} & 86.12 \se{.09} & 82.34 \se{.09} & \best{86.99 \se{.08}} & 86.58 \se{.08} \\
  Health & & & $ \cal R $ & 76.40 \se{.10} & 85.61 \se{.09} & 85.58 \se{.09} & 83.21 \se{.10} & 81.97 \se{.10} & 83.25 \se{.10} & \best{86.38 \se{.09}} & 85.94 \se{.09} & 79.67 \se{.09} & 85.81 \se{.09} & 85.66 \se{.09} \\
  & & & ${\cal F}_1$ & 76.14 \se{.11} & 85.69 \se{.09} & 85.60 \se{.09} & 83.32 \se{.10} & 81.94 \se{.10} & 83.40 \se{.10} & \best{86.40 \se{.09}} & 85.87 \se{.09} & 78.06 \se{.10} & 85.96 \se{.09} & 85.78 \se{.09} \\ [2pt] \hline
  Letter & 473 472 471 468 465 464 463 460 460 449 447 442 441 440 438 434 433 433 431 431 429 426 420 413 409 404 & 134 & $ \cal P $ & 93.24 \se{.01} & 93.31 \se{.01} & 93.09 \se{.01} & 92.88 \se{.01} & 91.48 \se{.01} & 93.28 \se{.01} & 93.17 \se{.01} & 93.28 \se{.01} & 93.25 \se{.01} & \best{93.83 \se{.01}} & 92.71 \se{.01} \\
  & & & $ \cal R $ & 93.09 \se{.01} & 93.18 \se{.01} & 92.94 \se{.01} & 92.70 \se{.01} & 91.07 \se{.02} & 93.13 \se{.01} & 93.03 \se{.01} & 93.15 \se{.01} & 93.10 \se{.01} & \best{93.72 \se{.01}} & 92.56 \se{.01} \\
  & & & ${\cal F}_1$ & 93.11 \se{.01} & 93.20 \se{.01} & 92.96 \se{.01} & 92.73 \se{.01} & 91.16 \se{.02} & 93.15 \se{.01} & 93.04 \se{.01} & 93.16 \se{.01} & 93.12 \se{.01} & \best{93.73 \se{.01}} & 92.57 \se{.01} \\ [2pt] \hline
  Sat & 969 935 858 376 367 312 & 103 & $ \cal P $ & 88.19 \se{.04} & \best{88.94 \se{.04}} & 88.57 \se{.04} & 87.77 \se{.04} & 86.00 \se{.04} & 88.20 \se{.04} & 88.93 \se{.04} & 87.86 \se{.04} & 88.36 \se{.04} & 88.57 \se{.04} & 88.67 \se{.04} \\
  Image & & & $ \cal R $ & 87.75 \se{.04} & \best{88.83 \se{.04}} & 88.37 \se{.04} & 87.44 \se{.04} & 82.03 \se{.05} & 87.95 \se{.04} & 88.57 \se{.04} & 87.48 \se{.04} & 87.02 \se{.04} & 88.28 \se{.04} & 88.53 \se{.04} \\
  & & & ${\cal F}_1$ & 87.47 \se{.04} & \best{88.83 \se{.04}} & 88.40 \se{.04} & 87.49 \se{.04} & 82.56 \se{.04} & 87.76 \se{.04} & 88.63 \se{.04} & 87.53 \se{.04} & 87.25 \se{.04} & 88.33 \se{.04} & 88.54 \se{.04} \\ [2pt] \hline
  Shuttle & 10022 1851 561 & 187 & $ \cal P $ & 99.73 \se{.01} & 99.86 \se{.01} & 99.83 \se{.01} & 99.75 \se{.01} & 99.69 \se{.01} & 99.63 \se{.01} & 99.82 \se{.01} & 99.76 \se{.01} & 99.82 \se{.01} & \best{99.92 \se{.00}} & 99.91 \se{.00} \\
  & & & $ \cal R $ & 99.73 \se{.01} & 99.86 \se{.01} & 99.83 \se{.01} & 99.75 \se{.01} & 99.69 \se{.01} & 99.63 \se{.01} & 99.82 \se{.01} & 99.76 \se{.01} & 99.82 \se{.01} & \best{99.91 \se{.00}} & 99.91 \se{.00} \\
  & & & ${\cal F}_1$ & 99.73 \se{.01} & 99.86 \se{.01} & 99.83 \se{.01} & 99.75 \se{.01} & 99.69 \se{.01} & 99.63 \se{.01} & 99.82 \se{.01} & 99.76 \se{.01} & 99.82 \se{.01} & \best{99.91 \se{.00}} & 99.91 \se{.00} \\ [2pt] \hline
  SVM & 208 104 40 & 13 & $ \cal P $ & 71.64 \se{.20} & 68.81 \se{.22} & 68.89 \se{.23} & 72.35 \se{.22} & 72.96 \se{.22} & 76.42 \se{.19} & 69.92 \se{.21} & 69.49 \se{.21} & 71.13 \se{.21} & \best{77.69 \se{.20}} & 75.90 \se{.20} \\
  Guide & & & $ \cal R $ & 62.84 \se{.19} & 67.93 \se{.21} & 67.12 \se{.21} & 69.11 \se{.22} & 69.53 \se{.21} & 72.03 \se{.20} & 68.69 \se{.21} & 68.05 \se{.21} & 67.46 \se{.20} & \best{75.94 \se{.20}} & 73.45 \se{.20} \\
  2 & & & ${\cal F}_1$ & 59.54 \se{.23} & 67.58 \se{.22} & 66.63 \se{.22} & 68.51 \se{.23} & 68.44 \se{.23} & 71.15 \se{.21} & 68.42 \se{.21} & 67.80 \se{.21} & 66.44 \se{.22} & \best{75.58 \se{.20}} & 72.67 \se{.21} \\ [2pt] \hline
  SVM & 99 96 90 79 66 62 & 20 & $ \cal P $ & \best{67.81 \se{.13}} & 66.67 \se{.13} & 66.02 \se{.13} & 65.14 \se{.13} & 66.32 \se{.13} & 63.55 \se{.13} & 66.69 \se{.13} & 66.65 \se{.13} & 66.56 \se{.13} & 67.23 \se{.13} & 67.14 \se{.13} \\
  Guide & & & $ \cal R $ & 64.09 \se{.12} & 65.32 \se{.13} & 64.67 \se{.12} & 63.34 \se{.13} & 64.61 \se{.12} & 61.53 \se{.13} & 65.25 \se{.13} & 65.45 \se{.12} & 64.38 \se{.13} & \best{65.62 \se{.13}} & 65.28 \se{.13} \\
  4 & & & ${\cal F}_1$ & 63.30 \se{.12} & 65.09 \se{.13} & 64.40 \se{.13} & 63.06 \se{.13} & 64.47 \se{.12} & 61.50 \se{.13} & 65.02 \se{.13} & 65.24 \se{.12} & 64.25 \se{.13} & \best{65.59 \se{.13}} & 65.14 \se{.13} \\ [2pt] \hline
  Vehicle & 169 168 163 150 & 49 & $ \cal P $ & 71.00 \se{.09} & 70.52 \se{.09} & 70.96 \se{.09} & 70.64 \se{.09} & 70.30 \se{.09} & \best{71.28 \se{.09}} & 71.02 \se{.09} & 71.07 \se{.09} & 71.06 \se{.09} & 69.82 \se{.09} & 70.28 \se{.10} \\
  & & & $ \cal R $ & 71.30 \se{.08} & 71.64 \se{.08} & 71.39 \se{.08} & 71.13 \se{.08} & 71.04 \se{.08} & \best{71.72 \se{.08}} & 71.46 \se{.08} & 71.62 \se{.08} & 71.44 \se{.08} & 71.25 \se{.08} & 71.53 \se{.08} \\
  & & & ${\cal F}_1$ & 70.93 \se{.09} & 70.73 \se{.09} & 70.99 \se{.09} & 70.70 \se{.08} & 70.43 \se{.09} & \best{71.31 \se{.09}} & 71.06 \se{.09} & 71.16 \se{.08} & 71.02 \se{.08} & 70.11 \se{.09} & 70.37 \se{.09} \\ [2pt] \hline
  Wine & 60 46 36 & 11 & $ \cal P $ & \best{97.41 \se{.08}} & 95.89 \se{.09} & 96.29 \se{.09} & 95.92 \se{.09} & 96.03 \se{.09} & 96.10 \se{.09} & 96.36 \se{.09} & 96.91 \se{.08} & 96.39 \se{.09} & 96.07 \se{.09} & 96.26 \se{.09} \\
  Recog & & & $ \cal R $ & \best{97.15 \se{.09}} & 95.26 \se{.10} & 95.82 \se{.10} & 95.36 \se{.10} & 95.50 \se{.10} & 95.57 \se{.10} & 95.89 \se{.10} & 96.57 \se{.09} & 95.90 \se{.10} & 95.52 \se{.10} & 95.74 \se{.10} \\
  & & & ${\cal F}_1$ & \best{97.12 \se{.09}} & 95.15 \se{.11} & 95.72 \se{.11} & 95.25 \se{.11} & 95.40 \se{.11} & 95.47 \se{.11} & 95.79 \se{.11} & 96.51 \se{.10} & 95.81 \se{.10} & 95.40 \se{.11} & 95.64 \se{.11} \\ [2pt] \hline
\end{tabular}%
\end{table}%

\subsection{Empirical Results}
\label{sec:empirical_multiclass}

We analyze 11 benchmark data sets to evaluate the performance of our proposed multi-class classifiers. The Dry Bean data set is taken from Kaggle and the E-Coli data set is taken from the UCI Machine Learning Repository. The rest of the data sets are available at LIBSVM (\url{https://www.csie.ntu.edu.tw/~cjlin/libsvmtools/datasets/}). For each data set, we form the training and the test sets (sizes are reported in Table~\ref{tab:benchmark_multiclass}) by randomly partitioning the data as before. At first, the observations from the smallest class are divided into two groups in nearly 3:1 ratio. Equal number of observations from each class are added to the smaller group to form the test set, while the training set is formed by the rest of the observations. This random partitioning is done 500 times, and the performance of different classifiers over these 500 partitions are reported in Table~\ref{tab:benchmark_multiclass}. 

In 5 out of these 11 data sets (Dry Bean, Letter, Shuttle, SVM Guide 2 and SVM Guide 4), the proposed OvO$+$ method has the best performance. Among 2 of these data sets (Dry Bean and Shuttle), the performance of the proposed OvR$+$ method is almost similar to OvO$+$, while in EColi data it outperform all other classifiers. Both the proposed methods also perform well on Fetal Health and Satimage data sets. Only in the case of Drugs data, our proposed methods, especially OvO$+$, has a poor performance.

\begin{figure}[h!]
	\centering
	\includegraphics[height=2.30in]{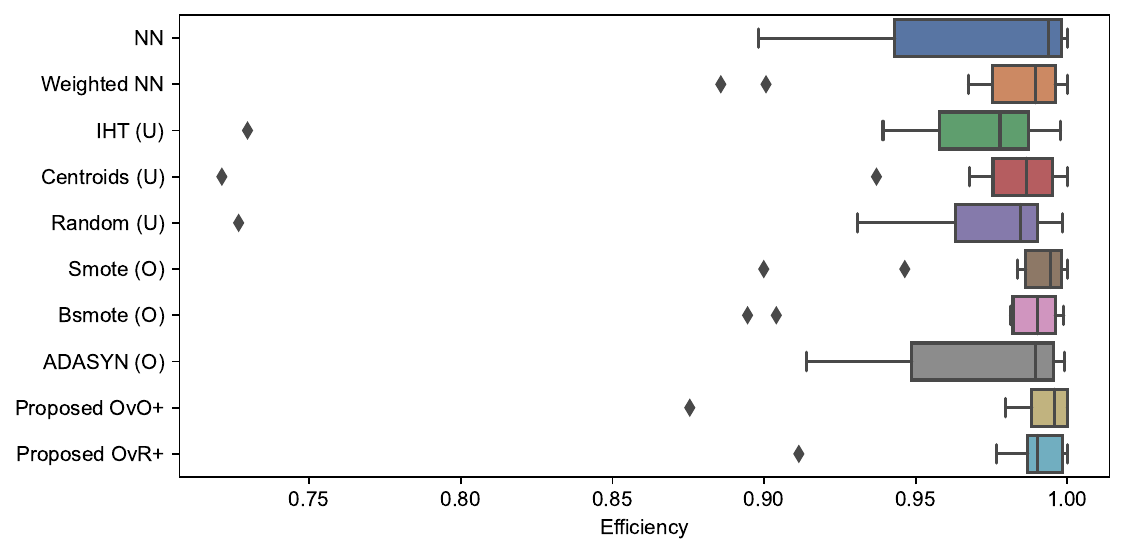}
	\caption{Boxplots of precision-efficiency of different classifiers on multi-class benchmark data sets. The undersampling methods are indicated by (U) and the oversampling methods are indicated by (O)\label{fig:prec_efficiency_multiclass}}
\end{figure}

To compare the overall performance of different methods on these 11 data sets, we construct the box plots of efficiency scores as before (see Figure~\ref{fig:prec_efficiency_multiclass}). These plots show that the overall performance of OvO$+$ is better than its competitors. The OvR$+$ method also had good overall performance. The precision-efficiency of OvR$+$ is slightly lower than that of OvO$+$, but its recall-efficiency and $F_1$-efficiency are comparable to OvO$+$. The performance of the undersampling methods are worse than the oversampling methods in the multi-class problems. 

\section{Concluding Remarks}
\label{sec:concluding_remarks}

We present a statistical method for nearest neighbor classification based on imbalanced training data. While the usual nearest neighbor classifier uses the same value of $k$ (the number of neighbors) for classification of all observations, here the choice of $k$ is case-specific, and it is obtained by maximizing a $p$-value-based evidence. Unlike the existing methods, our proposed classifiers do not need to use any adhoc weight functions, remove some observations, or add some pseudo observations for their implementation. So, the results are exactly reproducible. The proposed method is constructed under a probabilistic framework, and the resulting classifiers are consistent under mild assumptions. They can outperform the usual nearest neighbor classifier, the weighted nearest neighbor classifier and various data balancing algorithms in a wide variety of classification problems. Analyzing several simulated and real data sets, we have amply demonstrated this in this article.

\appendix
\section*{Appendix: Additional Results}
Here we provide a summary of the performance of all the classifiers considered in our numerical studies. For this summary, we consider the $F_1$-score of the classifiers and rank these values in each of the examples. Boxplots of these ranks are shown in Figure~\ref{fig:ranks_all}, separately for the two-class and the multi-class data sets. Note that a smaller rank indicates better performance by the classifier.

\begin{figure}[h!]
\centering
\small
\begin{tabular}{c}
(a) Two-class Problems \\ [2pt]
\includegraphics[height=2.5in]{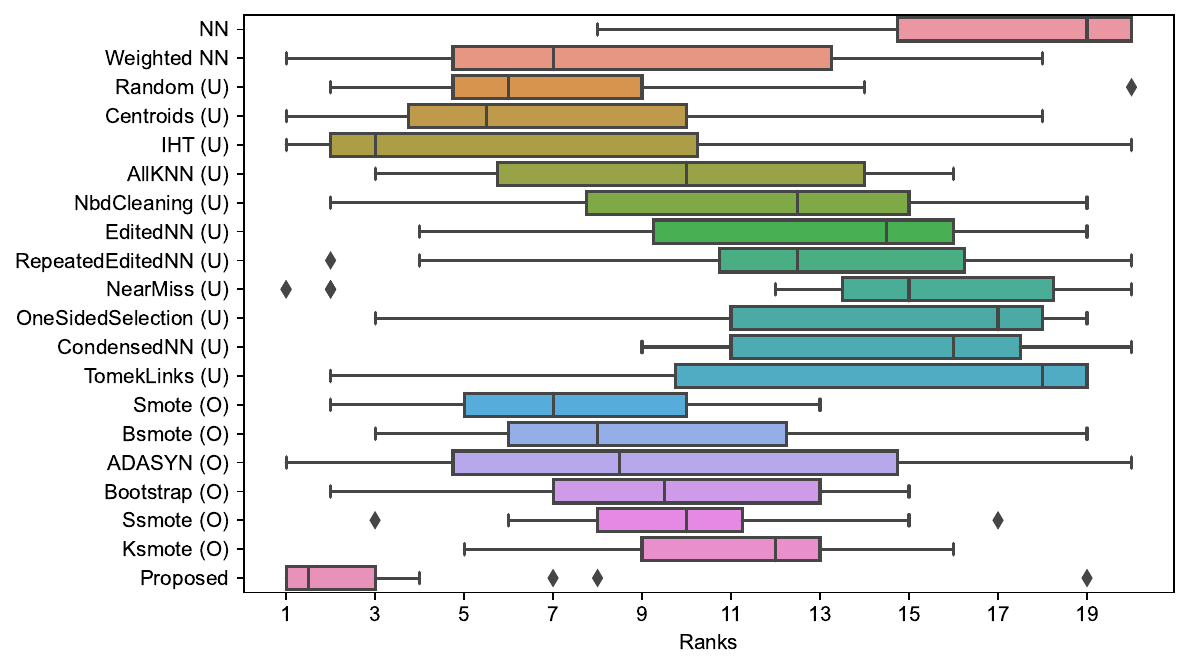} \\ [3pt]
(b) Multi-class Problems \\ [2pt]
\includegraphics[height=2.5in]{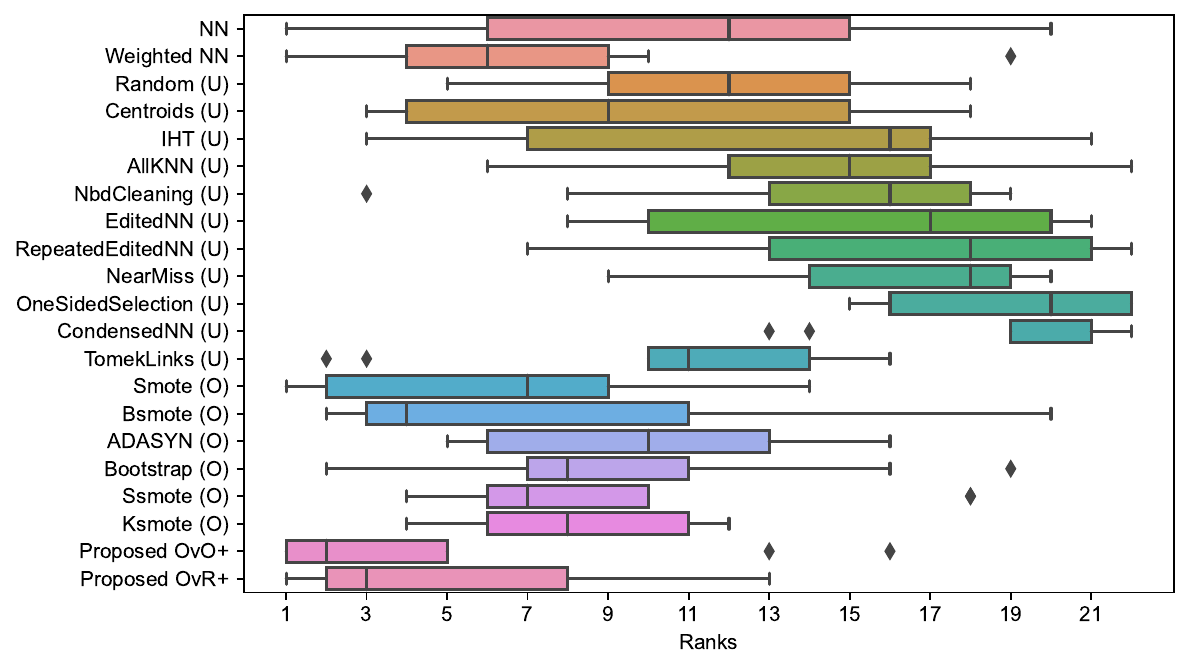}
\end{tabular}
\caption{Boxplots of ranks of $F_1$-scores for different classifiers in (a) the two-class problems and (b) the multi-class problems. The undersampling methods are indicated by (U) and the oversampling methods are indicated by (O).\label{fig:ranks_all}}
\end{figure}

From Figure~\ref{fig:ranks_all}, we can see that in the two-class problems, among the undersampling methods, random undersampling, Centroids and IHT have the best overall performance. Among the oversampling methods, Smote, Bsmote and ADASYN perform much better than the others. This is the reason for reporting the detailed results for only these methods in this paper. In the multi-class problem, random undersampling and Centroids are the best among the undersampling methods, and Smote and Bsmote are the best among the oversampling methods. But, the performance of IHT and ADASYN are much worse compared to the two-sample problems. However, in all the examples, the performance of our proposed methods is very impressive as can be seen from the figure. Specifically, for the two-class problems, our proposed method is among the best three methods (out of 20) most of the times. In the multi-class problems, the proposed OvO$+$ method is among the best five methods (out of 21) most of the times. The proposed OvR$+$ method has a slightly inconsistent performance. But, both these methods are clearly much better than the 19 other methods considered for comparison in our numerical studies.

\bibliographystyle{rss}
\bibliography{reference}
\end{document}